\documentclass[pra,aps,twocolumn,superscriptaddress,a4paper]{revtex4-1}
\usepackage[dvips]{graphicx}
\usepackage[USenglish]{babel}
\usepackage{amsmath,amssymb,amsthm,dsfont,url,braket,ctable,hyperref}
\usepackage[percent]{overpic}

\newcommand{\Tr}{\operatorname{Tr}}
\newcommand{\lin}[1]{\mathcal{L}(#1)}
\newcommand{\hilb}[1]{\mathcal{#1}}
\newcommand{\ch}[1]{\mathcal{#1}}
\newcommand{\N}[1]{\left|\!\left|{#1}\right|\!\right|}
\newcommand{\ketbra}[2]{\ket{#1}\!\!\bra{#2}}

\newcommand{\sgn}{\operatorname{sgn}}

\newtheorem{thm}{Theorem}
\newtheorem{cor}{Corollary}
\newtheorem{lmm}{Lemma}

\newcommand{\euset}[3]{S_{#1}^{#2}(#3)}

\newcommand{\thr}[2]{W \left( #1, #2 \right)}
\newcommand{\helstrom}[3]{H_{#1} \left(#2, #3 \right)}

\input{qcircuit}

\begin{document}

\title{Device-independent tests of quantum channels}

\author{Michele \surname{Dall'Arno}}

\email{cqtmda@nus.edu.sg}

\affiliation{Centre    for   Quantum    Technologies,   National
  University of Singapore, 3 Science Drive 2, 117543, Singapore}

\author{Sarah \surname{Brandsen}}

\email{sbrandse@caltech.edu}

\affiliation{Centre    for   Quantum    Technologies,   National
  University of Singapore, 3 Science Drive 2, 117543, Singapore}

\author{Francesco \surname{Buscemi}}

\email{buscemi@is.nagoya-u.ac.jp}

\affiliation{Graduate School of Information Science, Nagoya
  University, Chikusa-ku, Nagoya, 464-8601, Japan}

\date{\today}

\begin{abstract}
  We  develop  a  device-independent framework  for  testing
  quantum channels. That is, we falsify a hypothesis about a
  quantum  channel   based  only  on  an   observed  set  of
  input-output correlations. Formally,  the problem consists
  of  characterizing the  set  of input-output  correlations
  compatible with any arbitrary  given quantum channel.  For
  binary  (i.e.,  two  input symbols,  two  output  symbols)
  correlations,  we  show  that  extremal  correlations  are
  always achieved by  orthogonal encodings and measurements,
  irrespective  of  whether  or not  the  channel  preserves
  commutativity.   We further  provide  a full,  closed-form
  characterization of the sets of binary correlations in the
  case of:  i) any dihedrally-covariant qubit  channel (such
  as any Pauli and  amplitude-damping channels), and ii) any
  universally-covariant commutativity-preserving  channel in
  an arbitrary dimension (such as any erasure, depolarizing,
  universal cloning, and universal transposition channels).
\end{abstract}

\maketitle

\section{Introduction}
\label{sec:introduction}

Any  physical experiment  is based  upon the  observation of
correlations  among events  at various  points in  space and
time,  along  with  some assumptions  about  the  underlying
physics.   Naturally, in  order to  be operational  any such
assumption  must  have been  tested  as  a hypothesis  in  a
previous  experiment.   Ultimately,  to break  an  otherwise
circular   argument,   experiments  involving   no   further
assumptions  are  required  -- that  is,  device-independent
tests.

Formally, a hypothesis consists of a circuit~\cite{circuit},
which is usually  assumed to have a  global causal structure
(following  special relativity),  and its  components, which
are usually assumed  to be governed by  classical or quantum
theories and thus representable by channels.

Denoting  a hypothesis  (circuit)  by $\ch{X}$,  the set  of
correlations   compatible  with   $\ch{X}$  is   denoted  by
$\euset{}{}{\ch{X}}$.    Then,   hypothesis    $\ch{X}$   is
falsified,  along with  any other  hypothesis $\ch{Y}$  such
that  $\euset{}{}{\ch{Y}} \subseteq  \euset{}{}{\ch{X}}$, as
soon  as  the  observed   correlation  does  not  belong  to
$\euset{}{}{\ch{X}}$  (This  inclusion relation  induces  an
ordering  among  channels  which   is  reminiscent  of  that
introduced    by   Shannon~\cite{Sha58}    among   classical
channels).  Therefore,  from the theoretical  viewpoint, the
problem  of   falsifying  a   hypothesis  $\ch{X}$   can  be
recast~\cite{selftest}  as that  of  characterising the  set
$\euset{}{}{\ch{X}}$ of compatible correlations.

Since (discrete, memoryless) classical  channels are {\em by
  definition}    input-output   correlations    (conditional
probabilities), the characterisation of $\euset{}{}{\ch{X}}$
is trivial  in classical theory  as it is a  polytope easily
related  to the  correlation defining  the channel.   On the
contrary, the problem is far from trivial in quantum theory:
due  to  the  existence  of  superpositions  of  states  and
effects,  the  set   $\euset{}{}{\ch{X}}$  can  be  strictly
convex.

In this  work we  address the problem  of device-independent
tests    of   quantum    channels,    in   particular    the
characterization  of   the  set   $\euset{m}{n}{\ch{X}}$  of
$m$-inputs/$n$-outputs  correlations   $p_{j|i}$  obtainable
through an arbitrary given  channel $\ch{X}$, upon the input
of an  arbitrary preparation $\{ \rho_i  \}_{i=0}^{m-1}$ and
the   measurement   of   an   arbitrary   POVM   $\{   \pi_j
\}_{j=0}^{n-1}$, that is
\begin{align}
  \label{eq:chtest}
  p_{j|i} := \Tr[\ch{X}(\rho_i) \pi_j ] \quad = \quad
  \begin{aligned}
    \Qcircuit @C=4pt  @R=4pt {  \push{i \;}  & \prepareC{\rho_i}
      \cw & \gate{\ch{X}} & \measureD{\pi_j} & \cw & \push{\; j}
    }
  \end{aligned} \;.
\end{align}
The  analogous  problems   of  device-independent  tests  of
quantum states and measurements have been recently addressed
in Ref.~\cite{Dal17} and Ref.~\cite{DBBV16}, respectively.

An alternative  formulation for the problem  considered here
can be given  in terms of a ``game''  involving two parties:
an  experimenter, claiming  to  be able  to prepare  quantum
states, feed them through some quantum channel $\ch{X}$, and
then  perform measurements  on the  output, and  a skeptical
theoretician,   willing  to   trust  observed   correlations
only. If  the experimenter produces {\em  some} correlations
lying   outsides   of   $\euset{m}{n}{\ch{X}}$,   then   the
theoretician must conclude that the actual channel $\ch{X}'$
is not  worse than  $\ch{X}$ at producing  correlations, but
this is not sufficient  to support the experimenter's claim.
Indeed,  in   order  to   convince  the   theoretician,  the
experimenter   must    produce   {\em   the    entire   set}
$\euset{m}{n}{\ch{X}}$: in fact, it is sufficient to produce
a set  of correlations  whose convex  hull \textit{contains}
$\euset{m}{n}{\ch{X}}$. Then, the theoretician must conclude
that whatever  channel the  experimenter actually has  is at
least as good as $\ch{X}$ at producing correlations, and the
experimenter's claim is accepted.

It  is hence  clear that  the problem  of device-independent
tests  of quantum  channels induces  a preordering  relation
among quantum channels: $\ch{X}  \succeq \ch{Y}$ if and only
if  $\euset{m}{n}{\ch{X}}  \supseteq  \euset{m}{n}{\ch{Y}}$.
(The  order  also   depends  upon  $m$  and   $n$,  but  for
compactness we drop the indexes whenever they are clear from
the context).   In order to characterize  such preorder, for
any  given  channel $\ch{X}$,  we  need  to i)  provide  the
experimenter with all the states and measurements generating
the extremal correlations of $\euset{m}{n}{\ch{X}}$, and ii)
provide   the   theoretician   with   a   full   closed-form
characterization  of   the  set   $\euset{m}{n}{\ch{X}}$  of
compatible correlations.

As   a   preliminary  result,   we   find   that  the   sets
$\euset{m}{n}{\ch{X}}$  coincide   for  any  $d$-dimensional
unitary and  dephasing channels, for  any $d$, $m$,  and $n$
(this is an immediate consequence  of a remarkable result by
Frenkel and  Weiner~\cite{FW15}.) Upon considering  only the
binary case $m  = n = 2$,  our first result is  to show that
any correlation on the boundary of $\euset{2}{2}{\ch{X}}$ is
achieved by a pair of  commuting pure states -- irrespective
of whether  $\ch{X}$ is a  commutativity-preserving channel.
Then,   we    derive   the   {\em    complete   closed-form}
characterization of $\euset{2}{2}{\ch{X}}$ for: i) any given
dihedrally-covariant qubit channel,  including any Pauli and
amplitude-damping    channels;    and    ii)    any    given
universally-covariant    commutativity-preserving   channel,
including  any erasure,  depolarizing, universal  $1 \to  2$
cloning~\cite{Wer98},              and             universal
transposition~\cite{BDPS03} channels.

Upon  specifying $\ch{X}$  as  the $d$-dimensional  identity
channel   $\ch{I}_d$,    one   recovers   device-independent
dimension   tests   analogous    to   those   discussed   in
Refs.~\cite{GBHA10, HGMBAT12, ABCB11, DPGA12}, in which case
the  aforementioned   ordering  induced  by   the  inclusion
$\euset{m}{n}{\ch{I}_{d_0}}                        \subseteq
\euset{m}{n}{\ch{I}_{d_1}} \Leftrightarrow  d_0 \le  d_1$ is
of   course  total.   However,  the   completeness  of   our
characterization of $\euset{2}{2}{\ch{X}}$  implies that our
framework detects  {\em all} correlations  incompatible with
the given  hypothesis, unlike  Refs.~\cite{GBHA10, HGMBAT12,
  ABCB11, DPGA12,  CBB15} where  the set of  correlations is
tested only along an arbitrarily chosen direction.

Let  us  provide  a  preview of  some  consequences  of  our
results:
\begin{itemize}
\item Any Pauli  channel $\mathcal{P}^{\vec\lambda} : \rho
  \to \lambda_0 \rho +  \sum_{k=1}^3 \lambda_k \sigma_k \rho
  \sigma_k^\dagger$ is compatible with $p$ if and only if
  \begin{align*}
    \frac{|p_{1|1}  - p_{1|2}|}{1-|p_{1|1}  - p_{2|2}|}  \le
    \max_{k \in [1,3]} |2(\lambda_0 + \lambda_k) - 1|;
  \end{align*}
\item any amplitude-damping channel $\mathcal{A}^\lambda :
  \rho \to A_0 \rho A_0^\dagger + A_1 \rho A_1^\dagger$ with
  $A_0  = \ketbra{0}{0}  + \sqrt{\lambda}\ketbra{1}{1}$  and
  $A_1 = \sqrt{1-\lambda}  \ketbra{0}{1}$ is compatible with
  $p$ if and only if
  \begin{align*}
    \left(  \sqrt{p_{1|2}p_{2|1}}   -  \sqrt{p_{1|1}p_{2|2}}
    \right)^2 \le \lambda;
  \end{align*}
\item any  $d$-dimensional erasure channel  $\mathcal{E}_d :
  \rho \to  \lambda \rho \oplus (1-\lambda)  \Tr[\rho] \phi$
  for some pure  state $\phi$ is compatible with  $p$ if and
  only if
  \begin{align*}
    | p_{1|1} - p_{1|2}| \le \lambda;
  \end{align*}
\item     any    $d$-dimensional     depolarizing    channel
  $\mathcal{D}_d^\lambda   :  \rho   \to   \lambda  \rho   +
  (1-\lambda) \Tr[\rho]  \openone/d$ is compatible  with $p$
  if and only if
  \begin{equation*}
    \left\{\begin{split}
        &| p_{1|1} - p_{1|2}| \le \lambda,\\
        &         \frac{|p_{1|1}-p_{1|2}|}{1-|p_{1|1}-p_{2|2}|}                \le \frac{d\lambda}{2-2\lambda+d\lambda};
    \end{split}\right.
  \end{equation*}
\item  the  $d$-dimensional  universal  optimal  $1  \to  2$
  cloning~\cite{Wer98} channel $\mathcal{C}_d$ is compatible
  with $p$ if and only if
  \begin{align*}
      | p_{1|1} - p_{1|2}| \le \frac{d}{d+1};
  \end{align*}
\item     any      $d$-dimensional     universal     optimal
  transposition~\cite{BDPS03}  channel   $\mathcal{T}_d$  is
  compatible with $p$ if and only if
  \begin{equation*}
    \left\{\begin{split}
        | p_{1|1} - p_{1|2}| \le \frac1{d+1},\\
        \displaystyle  \frac{|p_{1|1}  -  p_{1|2}|}{1  -  |p_{1|1}-p_{2|2}|}  \le
        \frac13.
    \end{split}\right.
  \end{equation*}
\end{itemize}

This paper is structured as  follows.  We will introduce our
framework and  discuss the case  of unitary and  trace class
channels in Section~\ref{sec:general}.  For the binary case,
introduced  in Section~\ref{sec:binary},  we will  solve the
problem  for  any   qubit  dihedrally-covariant  channel  in
Section~\ref{sec:qubit},  and for  any arbitrary-dimensional
universally-covariant  commutativity-preserving  channel  in
Section~\ref{sec:covariant}.  In Section~\ref{sec:cartesian}
we will provide a  natural geometrical interpretation of our
results,   and  in   Section~\ref{sec:conclusion}  we   will
summarize our results and present further outlooks.

\section{General results}
\label{sec:general}

We will make use of  standard definitions and results in quantum
information theory~\cite{Wil11}. Since $\euset{m}{n}{\ch{X}}$ is
convex  for any  $n$  and $m$,  the  {\it hyperplane  separation
  theorem}~\cite{BV04,Bus12}     states    that     $p    \notin
\euset{m}{n}{\ch{X}}$ if and  only if there exists  an $m \times
n$ real matrix $w$ such that
\begin{align}
  \label{eq:hypsepthm}
  p^T \cdot w - \thr{\ch{X}}{w} > 0,
\end{align}
where $p^T \cdot w := \sum_{i,j} p_{j|i} w_{i,j} $, and
\begin{align}
  \label{eq:threshold}
  \thr{\ch{X}}{w} := \max_{q \in \euset{m}{n}{\ch{X}}} w^T \cdot
  q,
\end{align}
We     call    $w$     a    {\em     channel    witness}     and
$\thr{\ch{X}}{w}$  its  threshold value  for  channel
$\ch{X}$.

Although Eq.~\eqref{eq:hypsepthm} generally allows one to detect
{\em  some}   conditional  probability  distributions   $p$  not
belonging  to  $\euset{m}{n}{\mathcal{X}}$ for  any  arbitrarily
fixed witness $w$, here our aim is to detect {\em any} such $p$.
Direct application  of Eq.~\eqref{eq:hypsepthm}  is impractical,
as one would  need to consider {\em all} of  the infinitely many
witnesses $w$.  Notice however that Eq.~\eqref{eq:hypsepthm} can
be   rewritten  through   negation  by   stating  that   $p  \in
\euset{m}{n}{\ch{X}}$  if  and only  if  for  any $m  \times  n$
witness $w$ one has
\begin{align*}
  p^T \cdot w - \thr{\ch{X}}{w} \le 0,
\end{align*}
We then have our first preliminary result.

\begin{lmm}
  \label{thm:compatibility}
  A  channel  $\mathcal{X}  : \lin{\hilb{H}}  \to  \lin{\hilb{K}}$  is
  compatible with conditional probability distribution $p$ if and only
  if
  \begin{align}
    \label{eq:compatibility}
    \max_w \left[ p^T \cdot w - \thr{\ch{X}}{w} \right] \le 0.
  \end{align}
\end{lmm}

Let us start by considering an arbitrary $d$-dimensional unitary
channel $\ch{U}_d : \rho \to U \rho U^\dagger$, for some unitary
$U \in \lin{\hilb{H}}$  with $\dim\hilb{H} = d$.  If  $d \ge m$,
the maximization  in Eq.~\eqref{eq:threshold} is  trivial, since
the input labels $i \in [1,m]$  can all be encoded on orthogonal
states, so that  {\em any} $m \times  n$ conditional probability
distribution $q$ can in fact be  obtained.  However, if $d < m$,
the evaluation of the  witness threshold $\thr{\ch{U}_d}{w}$ for
any witness $w$  is far from obvious.   The solution immediately
follows  from  a  recent,   remarkable  result  by  Frenkel  and
Weiner~\cite{FW15}.   It turns  out that  $\thr{\ch{U}_d}{w}$ is
attained on  extremal conditional probability  distributions $q$
compatible with  the exchange  of a classical  $d$-level system,
namely, those  $q$ where $q_{j|i}  = 0$ or  $1$ for any  $i$ and
$j$, and  such that $q_{j|i} \neq  0$ for at most  $d$ different
values  of $j$.   Frenkel and  Weiner's result  hence guarantees
that the threshold $\thr{\ch{U}_d}{w}$ can be provided in closed
form since,  for any $m$  and $n$,  the number of  such extremal
classical  conditional probabilities  is finite,  i.e., the  set
$\euset{m}{n}{\ch{U}_d}$ is  a {\em polytope}.   Any probability
$p$ lying outside $\euset{m}{n}{\ch{U}_d}$  can thus be detected
by testing  the violation of Eq.~\eqref{eq:compatibility}  for a
{\em finite number} of witnesses $w$, corresponding to the faces
of the polytope.  Moreover,  the set $\euset{m}{n}{\ch{U}_d}$ of
distributions  compatible   with  any   $d$-dimensional  unitary
channel      $\ch{U}_d$     coincides      with     the      set
$\euset{m}{n}{\ch{F}_d^\lambda}$  of   distributions  compatible
with any  $d$-dimensional dephasing channel  $\ch{F}_d^\lambda :
\rho \to \lambda \rho + (1-\lambda) \sum_k \braket{k | \rho | k}
\ketbra{k}{k}$.

At  the opposite  end of  the unitary  channels, there  sit {\em
  trace-class  channels} $\ch{T}  :  \rho \to  \sigma$ for  some
arbitrary but fixed state $\sigma$. In this case, no information
about $i$ (the input label) can be communicated.  Of course, the
set  $\euset{m}{n}{\ch{T}}$ of  correlations achievable  through
any  trace-class  channel  $\ch{T}$   does  not  depend  on  the
particular  choice of  $\sigma$:  a  trace-class channel  simply
means that  no communication is available.   For any trace-class
channel  $\ch{T}$ and  any witness  $w$, it  immediately follows
that  the   threshold  $\thr{\mathcal{T}}{w}$  is   achieved  by
conditional  probabilities $q$  such that  $q_{j|i} =  1$ for  a
single value of $j$, and  therefore is given by $\thr{\ch{T}}{w}
=   \max_j  \sum_i   w_{i,j}$.   As   a  consequence,   the  set
$\euset{m}{n}{\ch{T}}$ is a polytope  with $n$ vertices, and any
probability  $p$  lying  outside $\euset{m}{n}{\ch{T}}$  can  be
detected       by      testing       the      violation       of
Eq.~\eqref{eq:compatibility}  for  a   {\em  finite}  number  of
witnesses $w$.

\section{Binary conditional probability distribution}
\label{sec:binary}

In the  remainder of this work  we will consider the  case where
$p$   is   a   binary   input-output   conditional   probability
distributions (i.e.   $m = n =  2$).

First,  we  show  that  it  suffices  to  consider  diagonal  or
anti-diagonal witnesses with positive entries summing up to one.
Indeed, for  any witness  $w$, the  witness $w'  := \alpha  (w +
\beta)$,  where   $\alpha  >  0$   and  $\beta$  is   such  that
$\beta_{i,j}$     is     independent      of     $j$,     leaves
Eq.~\eqref{eq:compatibility}   invariant  for   any  conditional
probability  distribution $p$  and channel  $\ch{X}$, since  $w'
\cdot p = \alpha (p^T \cdot w + \sum_i \beta_{i,1})$.

By taking $\beta_{i,j}  = -\min_k w_{i,k}$ for any  $i$ and $j$,
the witness  $w'$ is  diagonal, anti-diagonal,  or has  a single
non-null column.   We first consider the  latter case.  Clearly,
the maximum in Eq.~\eqref{eq:threshold}  is attained when $p$ is
a vertex of the polytope $\euset{2}{2}{\ch{T}}$ of probabilities
compatible  with  any   trace-type  channel  $\mathcal{T}$,  and
therefore Eq.~\eqref{eq:compatibility} is always verified.  Then
we consider  the case  of diagonal and  anti-diagonal witnesses.
By  taking  $\alpha^{-1}  =  \sum_i  |w_{i,1}  -  w_{i,2}|$  one
recovers the  normalization condition $\sum_{i,j} w_{i,j}  = 1$,
thus proving the statement.

Therefore, upon  denoting with $w^\pm(\omega)$ the  diagonal and
anti-diagonal witnesses given by
\begin{align*}
  w^+(\omega)  := \begin{pmatrix}  \frac{1+\omega}2 &  0 \\  0 &
    \frac{1-\omega}2 \end{pmatrix}, \quad w^-(\omega) :=
  \begin{pmatrix}  0 &  \frac{1+\omega}2  \\ \frac{1-\omega}2  &
    0 \end{pmatrix},
\end{align*}
where  $\omega \in  [-1,1]$, one  has the  following preliminary
result.

\begin{lmm}
  The maximum  in Eq.  \ref{eq:compatibility} is  attained for a
  diagonal or anti-diagonal witness, namely
  \begin{align*}
    & \max_w (p^T \cdot w - \thr{\ch{X}}{w}) \\ = & \max_{\omega
      \in    [-1,1]}    (p^T     \cdot    w^\pm(    \omega)    -
    \thr{\ch{X}}{w^\pm(\omega)}).
  \end{align*}
\end{lmm}

Any extremal distribution $q$ in Eq.~\eqref{eq:threshold} can be
represented  by states  $\rho_0$  and $\rho_1$  and  a POVM  $\{
\pi_0,  \pi_1  \}$  such   that  $q_{j|i}  =  \Tr[\ch{X}(\rho_i)
\pi_j]$.   Since $w^\pm(\omega)$  is diagonal  or anti-diagonal,
Eq.~\eqref{eq:threshold} represents  the maximum  probability of
success in the  discrimination of states $\{  \rho_0, \rho_1 \}$
with prior probabilities  given by the non-null  entries of $w$,
in the presence of noise $\ch{X}$, namely
\begin{align*}
  &    W(\mathcal{X},   w^\pm(\omega))    \\    =   &    \frac12
  \max_{\substack{\rho_0,  \rho_1\\\{\pi_0,  \pi_1  \}}}  \left[
    (1+\omega)    \Tr[\ch{X}(\rho_0)    \pi_0]   +    (1-\omega)
    \Tr[\ch{X}(\rho_1) \pi_1] \right].
\end{align*}
It is  a well-known fact~\cite{Hel76}  that the solution  of the
optimization problem  over POVMs is  given as a function  of the
Helstrom matrix defined as
\begin{align*}
  \helstrom{\omega}{\rho_0}{\rho_1} := \frac{1+\omega}2 \rho_0 -
  \frac{1-\omega}2 \rho_1,
\end{align*}
as follows
\begin{align}
  \label{eq:discrimination}
  W(\ch{X},w^\pm(\omega)) = \frac{1}{2}  \max_{ \rho_0, \rho_1 }
  \left[ 1 + \N{ \ch{X} \left( \helstrom{\omega}{\rho_0}{\rho_1}
      \right) }_1 \right],
\end{align}
where $\N{\cdot}_1$ denotes the operator $1$-norm.

It is easy  to see that without loss of  generality one can take
$\rho_0$ and $\rho_1$ such that  $[\rho_0, \rho_1] = 0$. Indeed,
let $\{ \ket{k}  \}$ be a basis of eigenvectors  of the Helstrom
matrix   $\helstrom{\omega}{\rho_0}{\rho_1}$.     The   complete
dephasing channel $\mathcal{F}_d^0$ on the basis $\{ \ket{k} \}$
is such that
\begin{align*}
  \helstrom{\omega}{\rho_0}{\rho_1}                            =
  \ch{F}_d^0(\helstrom{\omega}{\rho_0}{\rho_1})                =
  \helstrom{\omega}{\sigma_0}{\sigma_1},
\end{align*}
where $\sigma_i := \ch{F}_d^0(\rho_i)$ and therefore $[\sigma_0,
\sigma_1] = 0$.   By applying channel $\mathcal{X}$  we have the
following identity
\begin{align*}
  \mathcal{X}(\helstrom{\omega}{\rho_0}{\rho_1})               =
  \mathcal{X}(\helstrom{\omega}{\sigma_0}{\sigma_1})
\end{align*}
Therefore, the encoding $\{ \sigma_i \}$ performs as well as the
encoding $\{ \rho_i \}$, and  thus without loss of generality we
can  take  the  supremum in  Eq.~\eqref{eq:discrimination}  over
commuting encodings only.

Moreover, one  can see that  without loss of generality  one can
take  $\sigma_i$  to be  orthogonal  pure  states.  Indeed,  let
$\sigma_i  =  \sum_k  \mu_{k|i}  \ketbra{k}{k}$  be  a  spectral
decomposition of $\sigma_i$.  Due to  the convexity of the trace
norm we have
\begin{align*}
  &  \N{   \ch{X}  \left(  \helstrom{\omega}{\sigma_0}{\sigma_1}
    \right) }_1 \\ = & \N{ \sum_{k,l} \mu_{k|0} \mu_{l|1} \ch{X}
    \left(
      \helstrom{\omega}{\ketbra{k}{k}}{\ketbra{l}{l}} \right) }_1\\
  \le & \sum_{k, l}
  \mu_{k|0} \mu_{l|1} \N{ \ch{X} \left( \helstrom{\omega}{\ketbra{k}{k}}{\ketbra{l}{l}}\right) }_1 \\
  \le       &      \max_{k,l}       \N{      \ch{X}       \left(
      \helstrom{\omega}{\ketbra{k}{k}}{\ketbra{l}{l}}    \right)
  }_1.
\end{align*}
Then we have the following preliminary result.

\begin{lmm}
  \label{lmm:binary-threshold}
  The  maximum  in  Eq.~\eqref{eq:threshold}   is  given  by  an
  orthonormal pure encoding, namely
  \begin{align*}
    \thr{\mathcal{X}}{w^\pm(\omega)} :=
    \max_{\substack{\ket{\phi_0},         \ket{\phi_1}        \\
        \braket{\phi_1|\phi_0}   =  0}}   \frac12  \left[   1  +
      \N{\ch{X}    (\helstrom   {\omega}{\phi_0}{\phi_1}    )}_1
    \right],
  \end{align*}
  and  by  an  orthogonal  POVM such  that  $\pi_0$  is  the
  projector      on      the      positive      part      of
  $\helstrom{\omega}{\phi_0}{\phi_1}$ and  $\pi_1 = \openone
  - \pi_0$.
\end{lmm}
Here, for  any pure state  $\ket{\phi}$ we denote with  $\phi :=
\ketbra{\phi}{\phi}$  the  corresponding projector.

\section{Dihedrally covariant qubit channel}
\label{sec:qubit}

Let us  start with the  case where $\ch{X} :  \lin{\hilb{H}} \to
\lin{\hilb{K}}$  is  a  qubit   channel,  i.e.  $\dim\hilb{H}  =
\dim\hilb{K} = 2$. Since Pauli  matrices span the space of qubit
Hermitian operators, any qubit  state $\rho$ can be parametrized
in terms of Pauli matrices, i.e.
\begin{align}
  \label{eq:encoding}
  \rho =  \frac{1}{2}(\openone + \vec{\sigma}^T  \cdot \vec{x}),
  \qquad |\vec{x}|_2 \le 1,
\end{align}
where  $\vec{\sigma}  =  (\sigma_x, \sigma_y,  \sigma_z)^T$  and
$\vec{x}$  are the  vectors  of Pauli  matrices  and their  real
coefficients,  respectively.  Analogously,   any  qubit  channel
$\ch{X}$ can be parametrized in terms of Pauli matrices, i.e.
\begin{align*}
  \ch{X}(\rho) = \frac12 \left(  \openone + \vec{\sigma}^T \cdot
    (A \vec{x} + \vec{b}) \right),
\end{align*}
where $A_{i,j}  = \frac{1}{2} \Tr \left[  \sigma_i \ch{X} \left(
    \sigma_j \right) \right]$ and  $b_i = \frac{1}{2} \Tr \left[
  \sigma_i \ch{X} \left( \openone \right) \right]$.

With             such              a             parametrization
$\ch{X}(\helstrom{\omega}{\phi_0}{\phi_1})$   assumes   a   very
simple form given by
\begin{align*}
  \ch{X}(\helstrom{\omega}{\phi_0}{\phi_1}) = \frac{1}{2} \left[
    \omega  \openone  +  \left(   A  \vec{x}  +  \omega  \vec{b}
    \right)^T \!\cdot\vec{\sigma} \right],
\end{align*}
whose  eigenvalues are  $\frac{1}{2} \left(  \omega \pm  \left|A
    \vec{x}  +  \omega  \vec{b}\right|_2  \right)$.   Thus,  the
witness      threshold     $\thr{\ch{X}}{w^\pm(\omega)}$      in
Eq.~\eqref{eq:threshold}  can be  readily computed  by means  of
Lemma~\ref{lmm:binary-threshold} as
\begin{align*}
  \thr{\ch{X}}{w^\pm(\omega)}  =  \frac{1}{2}  \left[ 1  +  \max
    \left( |\omega|,  \max_{\substack{\vec{x} \\  \left| \vec{x}
          \right|_2  \le 1}}  \left|A \vec{x}  + \omega  \vec{b}
      \right|_2 \right) \right].
\end{align*}

Notice  that   this  expression  is  the   maximum  between  two
strategies. The first one is given  by the trivial POVM and thus
corresponds to trivial guessing.  The  second one can be further
simplified by means of the following substitutions. Let $A = V D
U$  be a  polar decomposition  of matrix  $A$ with  $U$ and  $V$
unitaries  and  $D$   diagonal  and  positive-semidefinite  with
eigenvalues  $\vec{d}$   (accordingly  $\vec{c}   :=  -V^\dagger
\vec{b}$). By unitary invariance of the $2$-norm one has
\begin{align*}
  \max_{\substack{\vec{x}  \\ \left|  \vec{x} \right|_2  \le 1}}
  \left|A    \vec{x}    +    \omega    \vec{b}    \right|_2    =
  \max_{\substack{\vec{x}  \\ \left|  \vec{x} \right|_2  \le 1}}
  \left|D \vec{x} - \omega \vec{c} \right|_2.
\end{align*}
By defining $\vec{y} := D \vec{x}$ one has 
\begin{align*}
  \max_{\substack{\vec{x}  \\ \left|  \vec{x} \right|_2  \le 1}}
  \left|D    \vec{x}    -    \omega    \vec{c}    \right|_2    =
  \max_{\substack{\vec{y},  \vec{z} \\  \left| D^{-1}  \vec{y} +
        \left( \openone - D^{-1}D  \right) \vec{z} \right|_2 \le
      1}} \left| \vec{y} - \omega \vec{c} \right|_2,
\end{align*}
where $(\cdot)^{-1}$ denotes the Moore-Penrose pseudoinverse. By
explicit computation it follows that $[D^{-1}]^T \left( \openone
  - D^{-1}D \right) = 0$, and therefore vectors $D^{-1} \vec{y}$
and $\left( \openone -  D^{-1}D \right) \vec{z}$ are orthogonal.
Then for any  optimal $\left( \vec{y}, \vec{z}  \right)$ one has
that $\left( \vec{y}, 0 \right)$  is also optimal, since $\left|
  D^{-1}  \vec{y} +  \left( \openone  - D^{-1}D  \right) \vec{z}
\right|_2 \ge  \left| D^{-1}  \vec{y} \right|_2$.   Therefore we
have
\begin{align}
  \label{eq:qubit-threshold}
  \thr{\ch{X}}{w^\pm(\omega)}  =  \frac{1}{2}  \left[ 1  +  \max
    \left( \omega, \Delta(\omega) \right) \right],
\end{align}
where
\begin{align}
  \label{eq:distance}
  \Delta(\omega) := \max_{\substack{\vec{y} \\
      \left| D^{-1}  \vec{y} \right|_2 \le 1}}  \left| \vec{y} -
    \omega \vec{c} \right|_2.
\end{align}

The  maximum  in   Eq.~\eqref{eq:distance}  is  a  quadratically
constrained quadratic optimization problem, which is known to be
NP-hard  in general.   However,  $\Delta(\omega)$  has a  simple
geometrical interpretation: it is the maximum Euclidean distance
of vector $\omega \vec{c}$  and ellipsoid $\left| D^{-1} \vec{y}
\right|_2 \le 1$.  This interpretation suggests symmetries under
which the optimization problem becomes feasible.  In particular,
we take  vector $\vec{c}$ to be  parallel to one of  the axis of
the ellipsoid  $\left| D^{-1}  \vec{y} \right|_2 \le  1$, namely
$c_1  =  c_2  =  0$   (up  to  irrelevant  permutations  of  the
computational basis).

This  configuration  corresponds  to a  $D_2$-covariant  channel
$\ch{X}$, where $D_2$ is the dihedral group of the symmetries of
a   line  segment,   consisting   of  two   reflections  and   a
$\pi$-rotation.  This configuration is depicted in
Fig.~\ref{fig:dihedral}.
\begin{figure}[hbt]
 \begin{overpic}[width=\columnwidth]{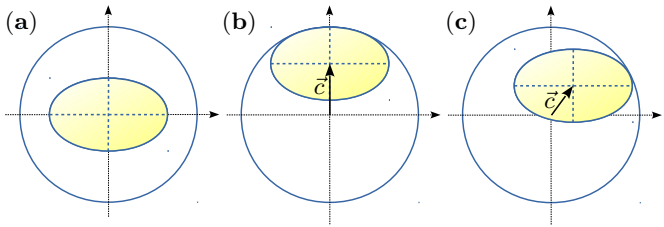}
   \put (0, 30) {$\bf (a)$}
   \put (33, 30) {$\bf (b)$}
   \put (47, 20) {$\vec{c}$}
   \put (67, 30) {$\bf (c)$}
   \put (82, 18) {$\vec{c}$}
 \end{overpic}
 \caption{Bloch-sphere  representation  of:  [{\bf  (a)},  {\bf
     (b)}] dihedrally  covariant channels $\ch{X}$  mapping the
   sphere into  an ellipsoid  {\bf (a)}  centered in  the Bloch
   sphere (e.g.  any  Pauli channel $\ch{P}^{\vec\lambda}$), or
   {\bf (b)} translated by a vector $\vec{c}$ which is parallel
   to one  of the  axis of the  ellipsoid (e.g.   any amplitude
   damping  channel $\ch{A}^\lambda$);  {\bf (c)}  non-dihedrally
   covariant channel  $\ch{X}$, as the ellipsoid  is translated
   by a  vector $\vec{c}$ which is  not parallel to any  of the
   axis of the ellipsoid.}
 \label{fig:dihedral}
\end{figure}
In  particular,  a qubit  channel  $\ch{X}$  is
$D_2$-covariant   if   and   only   if   there   exist   unitary
representations $U_k  \in \mathbb{R}^{3 \times 3}$  and $V_k \in
\mathbb{R}^{3 \times 3}$ of $D_2$ such that
\begin{align}
  \label{eq:d2-covariance}
  A U_k \vec{x} + \vec{b} = V_k(A \vec{x} + \vec{b}).
\end{align}
Up  to unitaries,  the  most general  unitary representation  of
$D_2$ in $\mathbb{R}^{3 \times 3}$ is given by
\begin{align*}
  W_1 = \sigma_z \oplus 1, \quad W_2 = -\sigma_z \oplus 1, \quad
  W_3 = - \openone \oplus 1,
\end{align*}
where  $W_1$   and  $W_2$  are   reflections  and  $W_3$   is  a
$\pi$-rotation. We take  $U_k := U^\dagger W_k U$ and  $V_k := V
W_k V^\dagger$. Then by explicit computation we have
\begin{align*}
  A U_k \vec{x}  + \vec{b} = V_k  A \vec{x} +
  \vec{b},
\end{align*}
where  we  used the  fact  that  $[D, W_k]  =  0$  for any  $k$.
Therefore,       $D_2$       covariance       expressed       by
Eq.~\eqref{eq:d2-covariance}  is equivalent  to the  requirement
$W_k \vec{c} = \vec{c}$, namely $c_1 = c_2 = 0$.

Under the  assumption of  $D_2$-covariance, we  take without
loss of generality  $d_2 \ge d_1$ and $c_3 \ge  0$.  If also
$c_3 = 0$,  we further take without loss  of generality $d_3
\ge d_2$.   Then, as  formally proved  in the  Appendix, the
maximum     Euclidean    distance     $\Delta(\omega)$    in
Eq.~\eqref{eq:distance} can be  explicitly computed, leading
to the following result.
\begin{lmm}
  \label{lmm:qubit-threshold}
  The  witness  threshold $\thr{\ch{X}}{w^\pm(\omega)}$  of  any
  qubit   $D_2$-covariant   channel   $\ch{X}$   is   given   by
  Eq.~\eqref{eq:qubit-threshold} where
  \begin{equation*}
    \Delta(\omega) =
    \left\{\begin{split}
        &d_2 \sqrt{1 + \frac{c_3^2 \omega^2}{d_2^2 - d_3^2}}, &\textrm{ if } |\omega| < \frac{d_2^2 - d_3^2}{d_3 c_3},\\
        &d_3 + c_3 |\omega|, &\textrm{ otherwise.}
    \end{split}\right.
\end{equation*}
  The optimal  encoding is given  by Eq.~\eqref{eq:encoding}
  with $\vec{x} = D^{-1} \vec{y}$ and
  \begin{align*}
    \vec{y} =
    \begin{cases}
      \left(0,        \pm        d_2        \sqrt{1        -
          \frac{c_3^2d_3^2\omega^2}{(d_3^2-d_2^2)^2}},
        \frac{c_3  d_3^2 \omega}{d_3^2  - d_2^2}\right)^T  &
      \textrm{        if        }        |\omega|        \le
      \frac{d_2^2-d_3^2}{d_3c_3}\\    \left(0,     0,    \pm
        d_3\right)^T & \textrm{ otherwise}.
    \end{cases}
  \end{align*}
\end{lmm}

Using             Lemma~\ref{lmm:qubit-threshold}            and
Lemma~\ref{thm:compatibility},      Eq.~\eqref{eq:compatibility}
becomes  the  maximum  over  $\omega$  of  the  minimum  of  two
functions.  The  maximum is attained  either in the  maxima $0$,
$\pm \omega_1$, or $ \pm 1$ of the two functions over the domain
$[-1, 1]$, where
\begin{align*}
  \omega_1     :=     \frac     {(d_2^2-d_3^2)(p_{1|1}-p_{2|2})}
  {c_3\sqrt{c_3^2d_2^2-(d_2^2-d_3^2)(p_{1|1}-p_{2|2})^2}},
\end{align*}
(the  limit should  be considered  if $c_3  = 0$),  or in  their
intersection $\pm\omega_2$ given by
\begin{equation*}
  \omega_2 :=
  \left\{\begin{split}
    &\sqrt{\frac{d_2^2(d_2^2  - d_3^2)}{d_2^2-d_3^2-d_2^2c_3^2}},
    &   \textrm{    if   }   (d_2^2-d_3^2)   >    d_2^2   c_3,\\
    &\frac{d_3}{1-c_3}, & \textrm{ otherwise.}
  \end{split}\right.
\end{equation*}
We can then state our  first main result, formally proved in
the   Appendix,   namely    a   complete   and   closed-form
characterization  of   the  set   $\euset{2}{2}{\ch{X}}$  of
conditional  probability distributions  compatible with  any
qubit $D_2$-covariant channel $\ch{X}$.

\begin{thm}
  \label{thm:qubit-compat}
  Any given  binary conditional probability distribution  $p$ is
  compatible  with  any   given  qubit  $D_2$-covariant  channel
  $\ch{X}$ if and only if
  \begin{align}
    \label{eq:qubit-compatibility}
    \max_{\omega   \in  \Omega}   (p^T  \cdot   w^\pm(\omega)  -
    \thr{\ch{X}}{w^\pm(\omega)}) \le 0,
  \end{align}
  where $\Omega  := \{ 0, \pm\omega_1,  \pm\omega_2, \pm1\} \cap
  [-1,1]$.
\end{thm}

As  applications   of  Theorem~\ref{thm:qubit-compat},   let  us
explicitly   characterize  the   sets   of  binary   conditional
probability distributions compatible  with two relevant examples
of    qubit   $D_2$-covariant    channels:    the   Pauli    and
amplitude-damping channels.

Any     Pauli     channel      can     be     written     as
$\mathcal{P}^{\vec\lambda}  :  \rho  \to  \lambda_0  \rho  +
\sum_{k=1}^3  \lambda_k   \sigma_k  \rho  \sigma_k^\dagger$,
where $\vec\sigma = (\sigma_x,  \sigma_y, \sigma_z)$ are the
Pauli  matrices.   One  has  that  $c_3  =  0$  and  $d_3  =
\max\limits_{k \in [1,3]} |2(\lambda_0 + \lambda_k) - 1| \ge
d_2$, thus $\omega_1 = \infty$  and $\omega_2 = d_3$ and the
maximum  in  Eq.~\eqref{eq:qubit-compatibility} is  attained
for   $\omega  =   \pm  \omega_2$.    Thus,  upon   applying
Theorem~\ref{thm:qubit-compat},   one   has  the   following
result.

\begin{cor}
  \label{cor:pauli}
  Any given  binary conditional probability distribution  $p$ is
  compatible  with the  Pauli channel  $\ch{P}^{\vec\lambda}$ if
  and only if
  \begin{align*}
    \frac{|p_{1|1}  -   p_{1|2}|}{1-|p_{1|1}  -   p_{2|2}|}  \le
    \max_{k \in [1,3]} |2(\lambda_0 + \lambda_k) - 1|.
  \end{align*}
\end{cor}

Any   amplitude-damping   channel    can   be   written   as
$\mathcal{A}^\lambda(\rho)    =   \sum_{k=0}^1    A_k   \rho
A_k^\dagger$,     where    $A_0     =    \ketbra{0}{0}     +
\sqrt{\lambda}\ketbra{1}{1}$  and  $A_1  =  \sqrt{1-\lambda}
\ketbra{0}{1}$.  As shown in the Appendix, one has that $c_3
=   1-\lambda$  and   $d_3  =   \lambda$,  $d_2   =  d_1   =
\sqrt{\lambda}$,     and     thus     the     maximum     in
Eq.~\eqref{eq:qubit-compatibility} is attained for $\omega =
\pm  \omega_1$  or $\omega  =  \pm1$.   Thus, upon  applying
Theorem~\ref{thm:qubit-compat},   one   has  the   following
result, formally proved in the Appendix.

\begin{cor}
  \label{cor:ampdamp}
  Any given  binary conditional probability distribution  $p$ is
  compatible with the amplitude-damping channel $\ch{A}^\lambda$
  if and only if
  \begin{align*}
    \left(    \sqrt{p_{1|2}p_{2|1}}   -    \sqrt{p_{1|1}p_{2|2}}
    \right)^2 \le \lambda.
  \end{align*}
\end{cor}

\section{Universally-covariant          commutativity-preserving
  channels}
\label{sec:covariant}

Let us  now move  to the arbitrary  dimensional case.   We trade
generality regarding the dimension  for generality regarding the
symmetry of the channel,  and assume {\em universal covariance}.
A  channel  $\ch{X}  :  \lin{\hilb{H}}  \to  \lin{\hilb{K}}$  is
universally  covariant  if  and  only  if  there  exist  unitary
representations   $U_g   \in   \lin{\hilb{H}}$  and   $V_g   \in
\lin{\hilb{K}}$ of the special unitary  group $SU(d)$ with $d :=
\dim{\hilb{H}}$,   such  that   for   every   state  $\rho   \in
\lin{\hilb{H}}$ one has
\begin{align}
  \label{eq:universal-covariance}
  \ch{X}(U_g \rho U_g^\dagger) = V_g \ch{X}(\rho) V_g^\dagger.
\end{align}

From universal covariance it  immediately follows that {\em any}
orthonormal   pure  encoding   attains  the   witness  threshold
$\thr{\ch{X}}{w^{\pm}(\omega)}$                               in
Eq.~\eqref{eq:discrimination}.  Indeed, for any orthonormal pure
states $\{ \phi_i \}$ let $U$ be the unitary such that $\phi_i =
U \ketbra{i}{i} U^\dagger$. Then one has
\begin{align*}
  & \N{\ch{X} \left( \helstrom{\omega}{\phi_0}{\phi_1} \right)}_1 \\
  = & \N{\ch{X} \left( \helstrom{\omega}{ U \ketbra{0}{0} U^\dagger}{ U \ketbra{1}{1} U^\dagger } \right)}_1 \\
  = & \N{V \ch{X} \left( \helstrom{\omega} {\ketbra{0}{0}} {\ketbra{1}{1}} \right)  V^\dagger }_1 \\
  =   &  \N{\ch{X}   \left(  \helstrom{\omega}   {\ketbra{0}{0}}
      {\ketbra{1}{1}} \right) }_1,
\end{align*}
where      the      second       equality      follows      from
Eq.~\eqref{eq:universal-covariance},  and  the  third  from  the
invariance of trace distance under unitary transformations. Then
we have the following result.

\begin{lmm}
  \label{lmm:universal-covariant}
  The  witness  threshold $\thr{\ch{X}}{w^\pm(\omega)}$  of  any
  universally covariant channel $\ch{X}$ is given by
  \begin{align}
    \label{eq:covariant-threshold}
    \thr{\ch{X}}{w^{\pm}(\omega)} = \frac12 \left[ 1 + \N{\ch{X}
        (\helstrom{\omega}{\ketbra{0}{0}}{  \ketbra{1}{1}})  }_1
    \right].
  \end{align}
  The optimal encoding  is given by any  pair of orthonormal
  pure states.
\end{lmm}

Equation~\eqref{eq:covariant-threshold} has  a simple dependence
on  $w$  in the  case  when  channel $\ch{X}$  is  commutativity
preserving, i.e. $[\ch{X}(\rho_0), \ch{X}(\rho_1)] = 0$ whenever
$[\ch{\rho}_0,  \ch{\rho}_1] =  0$. Notice  that it  suffices to
check  commutativity  preservation  for pure  states,  indeed  a
channel  $\ch{X}$ is  commutativity  preserving if  and only  if
$[\ch{X}(\phi_0),     \ch{X}(\phi_1)]      =     0$     whenever
$\braket{\phi_1|\phi_0}  =   0$.   Necessity  is   trivial,  and
sufficiency  follows by  assuming $[\ch{\rho}_0,  \ch{\rho}_1] =
0$, and  considering a  simultaneous spectral  decompositions of
$\rho_0  = \sum_k  \mu_k  \phi_k$ and  $\rho_1  := \sum_j  \nu_j
\phi_j$. Then one has
\begin{align*}
  \left[  \ch{X}(\rho_0), \ch{X}(\rho_1)  \right] &=  \sum_{k,l}
  \mu_k \nu_l \left[ \ch{X}(\phi_k), \ch{X}(\phi_l) \right]\\ &=
  0,
\end{align*}
where   the  last   inequality  follows   from  the   fact  that
$\braket{\phi_l|\phi_k}  =  \delta_{k,l}$.   For  a  universally
covariant   channel  $\ch{X}$,   it  immediately   follows  from
Eq.~\eqref{eq:universal-covariance}  that it  suffices to  check
commutativity preservation  for an arbitrary pair  of orthogonal
pure states.

In this case $\ch{X}(\ketbra{0}{0})$ and $\ch{X}(\ketbra{1}{1})$
admit a common basis of eigenvectors $\{ \ket{k} \}$, and thus a
spectral     decomposition     of    the     Helstrom     matrix
$\ch{X}(\helstrom{\omega}{\ketbra{0}{0}}{   \ketbra{1}{1}})$  is
given by
\begin{align*}
  \ch{X}(\helstrom{\omega}{\ketbra{0}{0}}{   \ketbra{1}{1}})   =
  \sum_k (\alpha_k \omega + \beta_k) \ketbra{k}{k},
\end{align*}
where   $\alpha_k$   and   $\beta_k$  are   the   half-sum   and
half-difference     of     the    $k$-th     eigenvectors     of
$\ch{X}(\ketbra{0}{0})$       and       $\ch{X}(\ketbra{1}{1})$,
respectively.     Therefore   Eq.~\eqref{eq:covariant-threshold}
becomes
\begin{align*}
  \thr{\ch{X}}{w^\pm(\omega)}  =  \frac12   \left(  1  +  \sum_k
    |\alpha_k \omega + \beta_k| \right).
\end{align*}

Then, the  optimization problem  in Eq.~\eqref{eq:compatibility}
becomes piece-wise linear,  thus the maximum is  attained on the
intersections of the piece-wise components given by $\gamma_k :=
\beta_k/\alpha_k$ when  such values belongs to  the domain $[-1,
1]$, or  on its extrema.   We can  then provide our  second main
result, namely  a complete  closed-form characterization  of the
set    $\euset{2}{2}{\ch{X}}$    of   conditional    probability
distributions   compatible    with   any   arbitrary-dimensional
universally-covariant commutativity-preserving channel $\ch{X}$.

\begin{thm}
  \label{thm:covariant-compatibility}
  Any given  binary conditional probability distribution  $p$ is
  compatible     with     any    given     arbitrary-dimensional
  universally-covariant     commutativity-preserving     channel
  $\ch{X}$ if and only if
  \begin{align*}
    \begin{cases}
      |p_{1|1} - p_{1|2}| \le \sum_k |\beta_k|,\\
      |p_{1|1} -  p_{1|2}| \le \N{ \ch{X}(  \helstrom {\gamma_k}
        {\ketbra{0}{0}} {\ketbra{1}{1}}) }_1 - \gamma_k |p_{1|1}
      - p_{2|2}|,
    \end{cases}
  \end{align*}
  for any $k$ such that $\gamma_k \in [-1, 1]$.
\end{thm}

As  applications  of  Theorem~\ref{thm:covariant-compatibility},
let  us explicitly  compute the  binary conditional  probability
distributions   compatible  with   any  erasure,   depolarizing,
universal  optimal  $1 \to  2$  cloning,  and universal  optimal
transposition  channels.  As   discussed  before,  commutativity
preservation  can  be  immediately  verified for  all  of  these
channels   by  checking   that  $\left[   \ch{X}(\ketbra{0}{0}),
  \ch{X}(\ketbra{1}{1}) \right] = 0$.

Any erasure channel  can be written as  $\ch{E}_d^\lambda : \rho
\to \lambda \rho \oplus (1-\lambda)  \phi$, where $\phi$ is some
pure  state.   One  can  compute  that  $\vec{\alpha}  =  \left(
  \frac{\lambda}2,  \frac{\lambda}2,  0  \times  d-2,  1-\lambda
\right   )$   and   $\vec{\beta}   =   \left(   \frac{\lambda}2,
  -\frac{\lambda}2, 0  \times d-1  \right)$, thus  upon applying
Theorem~\ref{thm:covariant-compatibility} one  has the following
Corollary.

\begin{cor}
  \label{cor:erasure}
  Any given  binary conditional probability distribution  $p$ is
  compatible with the erasure  channel $\ch{E}_d^\lambda$ if and
  only if
  \begin{align*}
    | p_{1|1} - p_{1|2} | \le \lambda.  
  \end{align*}
\end{cor}

Any depolarizing  channel can be written  as $\ch{D}_d^\lambda :
\rho \to \lambda \rho  + (1-\lambda) \frac{\openone}d$.  One can
compute   that    $\vec{\alpha}   =    \left(\frac{\lambda}2   +
  \frac{1-\lambda}d  \times  2,   \frac{1-\lambda}d  \times  d-2
\right   )$   and   $\vec{\beta}  =   \left(   -\frac{\lambda}2,
  \frac{\lambda}2,  0 \times  d-2 \right)$,  thus upon  applying
Theorem~\ref{thm:covariant-compatibility} one  has the following
Corollary.

\begin{cor}
  \label{cor:depolarizing}
  Any given  binary conditional probability distribution  $p$ is
  compatible with the depolarizing channel $\ch{D}_d^\lambda$ if
  and only if
  \begin{equation*}
    \left\{\begin{split}
        &|p_{1|1} - p_{1|2}| \le \lambda,\\
        &         \frac{|p_{1|1}-p_{1|2}|}{1-|p_{1|1}-p_{2|2}|}                \le \frac{d\lambda}{2-2\lambda+d\lambda}.
    \end{split}\right.
  \end{equation*}
\end{cor}

The  universal optimal $1  \to  2$  cloning channel  can  be written  as
$\ch{C}_d^\lambda  :  \rho  \to \frac2{d+1}  P_S  (\rho  \otimes
\openone) P_S$. By explicit computation one has
\begin{align*}
  \ch{C}_d(\ketbra{i}{i}) = \frac{1}{2(d+1)} \sum_k (\ket{k,i} +
  \ket{i,k}) (\bra{k,i} + \bra{i,k}),
\end{align*}
and             therefore             $[\ch{C}_d(\ketbra{0}{0}),
\ch{C}_d(\ketbra{1}{1})] = 0$, thus the universal optimal $1 \to
2$  cloning $\ch{C}_d$  is a  commutativity preserving  channel.
One can compute that $\vec{\alpha}  = \left(\frac1{d + 1} \times
  3, \frac{1}{2(d+1)}  \times 2(d-2)\right)$ and  $\vec{\beta} =
\left(  -\frac1{d+1}, \frac1{d+1},  0,  - \frac1{2(d+1)}  \times
  d-2, \frac1{2(d+1)}  \times d-2  \right)$, thus  upon applying
Theorem~\ref{thm:covariant-compatibility} one  has the following
Corollary.

\begin{cor}
  \label{cor:cloning}
  Any given  binary conditional probability distribution  $p$ is
  compatible  with  the  universal optimal  $1 \to  2$  cloning  channel
  $\ch{C}_d$ if and only if
  \begin{align*}
    |p_{1|1} - p_{1|2}| \le  \frac{d}{d+1}.
  \end{align*}
\end{cor}

The universal transposition channel  can be written as $\ch{T}_d
: \rho \to  \frac1{d+1} \left( \rho^T +  \openone \right)$.  One
can compute that $\vec{\alpha} = \left( \frac3{2(d+1)} \times 2,
  \frac1{d+1}  \times d-2  \right)$  and  $\vec{\beta} =  \left(
  \frac1{2(d+1)}, -\frac1{2(d+1)},  0 \times d-2  \right)$, thus
upon applying  Theorem~\ref{thm:covariant-compatibility} one has
the following Corollary.

\begin{cor}
  \label{cor:transposition}
  Any given  binary conditional probability distribution  $p$ is
  compatible with the universal transposition channel $\ch{T}_d$
  if and only if
  \begin{equation*}
    \left\{\begin{split}
      |p_{1|1} - p_{1|2}| \le \frac1{d+1},\\
      \frac{|p_{1|1}-p_{1|2}|}{1-|p_{1|1}-p_{2|2}|} \le \frac13.
    \end{split}\right.
  \end{equation*}
\end{cor}

The  results of  Corollaries~\ref{cor:pauli}, \ref{cor:ampdamp},
\ref{cor:erasure},   \ref{cor:depolarizing},  \ref{cor:cloning},
and      \ref{cor:transposition}      are     summarized      in
Table~\ref{tab:results}.

\begin{center}
  \begin{table}[h]
    \begin{tabular}{| >{$}c<{$} | >{$}c<{$} |}
      \hline
      \ch{X} & p \in \euset{2}{2}{\ch{X}} \\ 
      \hline\hline
      \ch{P}^{\vec\lambda} &  |p_{1|1}-p_{1|2}| \leq \max\limits_{k \in [1,3]} |2 (\lambda_0 + \lambda_k)-1| \\
      \hline
      \ch{A}^\lambda & (\sqrt{p_{1|2}p_{2|1}}-\sqrt{p_{1|1}p_{2|2}})^{2} \leq \lambda \\
      \hline
      \ch{E}_d^\lambda & |p_{1|1} - p_{1|2}| \le \lambda \\
      \hline
      \ch{D}_d^\lambda & \begin{cases} |p_{1|1}-p_{1|2}| \leq \lambda \\ \frac{|p_{1|1}-p_{1|2}|}{1-|p_{1|1}-p_{2|2}|} \le \frac{d\lambda}{2-2\lambda+d\lambda} \end{cases} \\
      \hline
      \ch{C}_d & |p_{1|1}-p_{1|2}| \leq \frac{d}{d+1} \\
      \hline
      \ch{T}_d &  \begin{cases} |p_{1|1}-p_{1|2}| \leq  \frac{1}{d+1} \\ \frac{|p_{1|1}-p_{1|2}|}{1-|p_{1|1}-p_{2|2}|} \le \frac{1}{3} \end{cases}  \\
      \hline
    \end{tabular}
    \caption{Complete closed-form characterization of the set $\euset{2}{2}{\ch{X}}$ of binary conditional probability distributions compatible with channel $\ch{X}$, for $\ch{X}$ given by the Pauli channel $\ch{P}^{\vec\lambda}$, the amplitude damping channel $\ch{A}^\lambda$, the erasure channel $\ch{E}_{d}^{\lambda}$, the depolarizing channel $\ch{D}_{d}^{\lambda}$, the universal $1 \to 2$ cloning channel $\ch{C}_{d}$, and the universal transposer $\ch{T}_{d}$,  as given by Corollaries~\ref{cor:pauli}, \ref{cor:ampdamp}, \ref{cor:erasure},   \ref{cor:depolarizing},  \ref{cor:cloning}, and      \ref{cor:transposition}, respectively.}
    \label{tab:results}
  \end{table}
\end{center}

\section{Cartesian representation}
\label{sec:cartesian}

In this Section  we provide a geometrical  interpretation of our
results.   Binary  conditional   probability  distributions  are
represented  by $2  \times  2$  real matrices,  so  they can  be
regarded  as vectors  in  $\mathbb{R}^4$.  However,  due to  the
normalization constraint $\sum_j p_{j|i} =  1$ for any $i$, they
all lie in a bidimensional affine subspace.  A natural Cartesian
parametrization of such a subspace is given by
\begin{align}
  \label{eq:cartesian}
  p_{j|i} = p(x,y) = \frac{1}{2} \left[ \begin{pmatrix} 1 & 1 \\
      1 &  1 \end{pmatrix}  + x  \begin{pmatrix} 1 &  -1 \\  1 &
      -1  \end{pmatrix} +  y  \begin{pmatrix}  1 &  -1  \\ -1  &
      1 \end{pmatrix} \right],
\end{align}
and binary conditional probability distributions form the square
$|x \pm  y| \le  1$, whose $4$  vertices are  the right-stochastic
matrices with all entries equal to $0$ or $1$.

As it is clear from Eq.~\eqref{eq:cartesian}:
\begin{itemize}
\item  a  permutation  of  the states  $\{  \rho_0,  \rho_1  \}$
  corresponds to the transformation $(x,y) \to (x,-y)$;
\item  a  permutation  of  the  effects  $\{  \pi_0,  \pi_1  \}$
  corresponds to the transformation $(x,y) \to (-x, -y)$;
\item a  permutation of  the states $\{  \rho_0, \rho_1  \}$ and
  effects $\{ \pi_0, \pi_1 \}$ corresponds to the transformation
  $(x,y)   \to   (-x,   y)$.
\end{itemize}
Therefore,    for     any    channel    $\ch{X}$,     the    set
$\euset{2}{2}{\ch{X}}$   of   binary   conditional   probability
distributions   compatible  with   $\ch{X}$  is   symmetric  for
reflections  around   the  $x$   or  $y$   axes  (i.e.,   it  is
$D_2$-covariant).

As   a   consequence  of   our   previous   results,  the   sets
$\euset{2}{2}{\ch{U}_d}$ and $\euset{2}{2}{\ch{F}_d^\lambda}$ of
conditional  probability   distributions  compatible   with  any
unitary and dephasing channels $\ch{U}_d$ and $\ch{F}_d^\lambda$
coincide with the  square $|x \pm y|  \le 1$, for any  $d$ and any
$\lambda$.   The   set  $\euset{2}{2}{\ch{T}}$   of  conditional
probability  distributions   compatible  with   any  trace-class
channel $\ch{T}$ coincide with the segment $x \in [-1, 1]$, $y =
0$.

With the  parametrization in Eq.~\eqref{eq:cartesian},  the sets
of binary conditional  probability distributions compatible with
any Pauli,  amplitude-damping, erasure,  depolarizing, universal
$1 \to 2$ cloning, and universal transposition channels as given
by        Corollaries~\ref{cor:pauli},        \ref{cor:ampdamp},
\ref{cor:erasure},   \ref{cor:depolarizing},  \ref{cor:cloning},
and   \ref{cor:transposition}   respectively,   are   given   in
Table~\ref{tab:cartesian} and depicted in Fig.~\ref{fig:binary}.

\begin{center}
  \begin{table}[h]
    \begin{tabular}{| >{$}c<{$} | >{$}c<{$} |}
      \hline
      \ch{X} & p(x,y) \in \euset{2}{2}{\ch{X}}   \\ 
      \hline\hline
      \ch{P}^{\vec\lambda} &    |y| \leq \max\limits_{k \in [1,3]} |2 (\lambda_0 + \lambda_k)-1| \\
      \hline
      \ch{A}^\lambda & \frac14     \left(\sqrt{1-2y-x^2+y^2}-\sqrt{1+2y-x^2+y^2} \right)^{2} \leq \lambda \\
      \hline
      \ch{E}_d^\lambda & |y| \le \lambda\\
      \hline
      \ch{D}_d^\lambda & \begin{cases} |y| \leq \lambda \\ \frac{|y|}{1-|x|} \le \frac{d\lambda}{2-2\lambda+d\lambda} \end{cases} \\
      \hline
      \ch{C}_d & |y| \leq \frac{d}{d+1} \\
      \hline
      \ch{T}_d & \begin{cases} |y| \leq  \frac{1}{d+1} \\ \frac{|y|}{1-|x|} \le \frac{1}{3} \end{cases} \\
      \hline
    \end{tabular}
    \caption{Cartesian parametrization of the set $\euset{2}{2}{\ch{X}}$ of binary conditional probability distributions compatible with channel $\ch{X}$, for $\ch{X}$ given by the Pauli channel $\ch{P}^{\vec\lambda}$, the amplitude damping channel $\ch{A}^\lambda$, the erasure channel $\ch{E}_{d}^{\lambda}$, the depolarizing channel $\ch{D}_{d}^{\lambda}$, the universal $1 \to 2$ cloning channel 
$\ch{C}_{d}$, and the universal transposer $\ch{T}_{d}$.}
    \label{tab:cartesian}
  \end{table}
\end{center}

\begin{figure}[h]
 \begin{overpic}[width=\columnwidth]{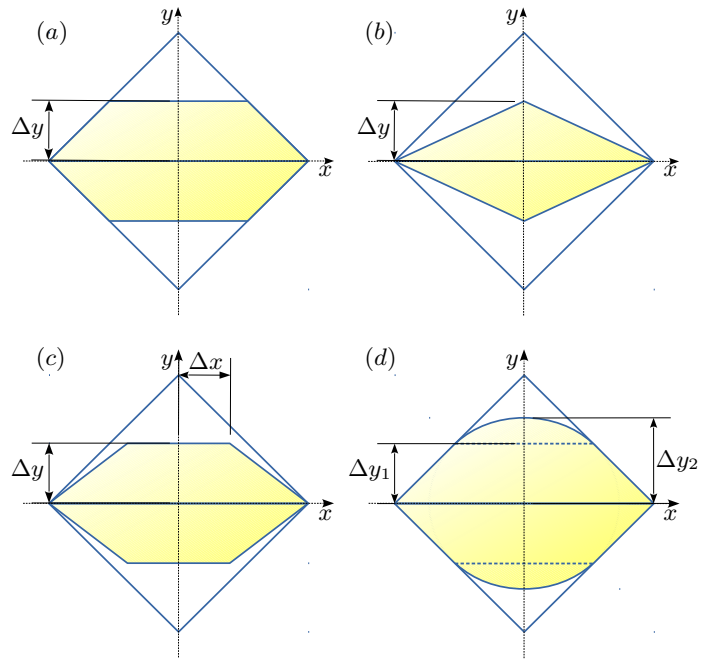}
   \put (2,95) {$(a)$}
   \put (52,95) {$(b)$}
   \put (2,45) {$(c)$}
   \put (52,45) {$(d)$}
   \put (45,74) {$x$}
   \put (97.5,74) {$x$}
   \put (45,21.2) {$x$}
   \put (97.5,21.2) {$x$}
   \put (21,98) {$y$}
   \put (73.5,98) {$y$}
   \put (21,45) {$y$}
   \put (73.5,45) {$y$}
   \put (-1.8, 80.5) {$\Delta y$}
   \put (50.8, 80.5) {$\Delta y$}
   \put (-1.8, 28) {$\Delta y$}
   \put (25.2, 44.5) {$\Delta x$}
   \put (49.5, 28) {$\Delta y_1$}
   \put (96.3, 29.5) {$\Delta y_2$}
 \end{overpic}
 \caption{Cartesian  representation  of  the  space  of  binary
   conditional probability distributions  $p$.  The outer white
   square  denotes  the  polytope  of  all  binary  conditional
   probability distributions.  The  inner yellow region denotes
   the sets  $\euset{2}{2}{\ch{X}}$ of  conditional probability
   distributions  compatible  with:  \textbf{(a)}  the  erasure
   channel  $\ch{X}  =  \ch{E}_d^\lambda$   (for  $\Delta  y  =
   \lambda$)  and  the  universal  optimal $1  \to  2$  cloning
   channel   $\ch{X}    =   \ch{C}_d$   (for   $\Delta    y   =
   \frac{d}{d+1}$);  \textbf{(b)} the  Pauli channel  $\ch{X} =
   \ch{P}^{\vec\lambda}$ (for  $\Delta y  = \max_{k  \in [1,3]}
   |2(\lambda_0    +    \lambda_k)-1|$);    \textbf{(c)}    the
   depolarizing  channel   $\ch{X}  =   \ch{D}_d^\lambda$  (for
   $\Delta  x =  \frac{d-2}{d}(1 -  \lambda)$ and  $\Delta y  =
   \lambda$)  and the  universal optimal  transposition channel
   $\ch{T}_d$ (for $\Delta x = \frac{d-2}{d+1}$ and $\Delta y =
   \frac1{d+1}$);  \textbf{(d)}  the amplitude-damping  channel
   $\ch{A}^\lambda$ (for  $\Delta y_1 = \lambda$  and $\Delta
   y_2 = \sqrt{\lambda}$).}
 \label{fig:binary}
\end{figure}

\section{Conclusions and outlook}
\label{sec:conclusion}

In this  work, we  developed a  device-independent framework
for testing quantum  channels.  The problem was  framed as a
game  involving  an experimenter,  claiming  to  be able  to
produce some quantum channel, and a theoretician, willing to
trust  observed  correlations  only.  The  optimal  strategy
consists  of  i)  all  the  input  states  and  measurements
generating the  extremal correlations that  the experimenter
needs   to    produce,   and   ii)   a    full   closed-form
characterization  of the  correlations  compatible with  the
claim,  that  the theoretician  needs  to  compare with  the
observed   correlations.    For  binary   correlations,   we
explicitly derived the optimal  strategy for the cases where
the claimed channel is a dihedrally-covariant qubit channel,
such  as any  Pauli  and amplitude-damping  channels, or  an
arbitrary-dimensional                  universally-covariant
commutativity-preserving  channel,  such   as  any  erasure,
depolarizing, universal cloning, and universal transposition
channels.

Natural generalisation  of our results include  relaxing the
restriction of binary correlations, that is $m = n = 2$, and
extending the characterization  of $\euset{m}{n}{\ch{X}}$ to
other  classes of  channels.  An  interesting generalisation
would consist of  letting the POVM $\{ \pi_y \}$  depend upon an
input not known during the preparation of $\{ \rho_x \}$, as
is the case  in quantum random access  codes.  Moreover, the
setup in  Eq.~\eqref{eq:chtest} could  be modified  to allow
for  entanglement alongside  $\ch{X}$, or  many parallel  or
sequential uses of channel $\ch{X}$.

We conclude  by remarking that our  results are particularly
suitable for  experimental implementation.  For  any channel
$\ch{X}$ an experimenter  claims to be able  to produce, our
framework  only requires  them  to  prepare orthogonal  pure
input states and perform orthogonal measurements in order to
fully    characterize   $\euset{2}{2}{\ch{X}}$    and   thus
device-independently test $\ch{X}$.

\section*{Data accessibility}

This work does not have any experimental data.

\section*{Competing interests}

The  authors declare  no  competing  interests.

\section*{Author's contributions}

All  authors  equally  contributed to  the  original  ideas,
analytical   derivations,   and   final  writing   of   this
manuscript, and gave final approval for publication.

\section*{Acknowledgements}

We  are grateful  to Alessandro  Bisio, Antonio  Ac\'in, Giacomo
Mauro D'Ariano,  and Vlatko Vedral for  valuable discussions and
suggestions.

\section*{Funding}

M.~D.  acknowledges  support from the Singapore  Ministry of
Education  Academic   Research  Fund   Tier  3   (Grant  No.
MOE2012-T3-1-009).  F.~B acknowledges  support from the JSPS
KAKENHI, No.  26247016.

\appendix

\section{Proofs}

In this Section we prove  those results reported in the previous
Sections for which the proof, being lengthy and not particularly
insightful, had only been  outlined. The numbering of statements
follows that of the previous Sections.

\setcounter{lmm}{3}

\begin{lmm}
  The  witness  threshold $\thr{\ch{X}}{w^\pm(\omega)}$  of  any
  qubit   $D_2$-covariant   channel   $\ch{X}$   is   given   by
  Eq.~\eqref{eq:qubit-threshold} where
  \begin{align*}
    \Delta(\omega) =
    \begin{cases}
      d_2 \sqrt{1 + \frac{c_3^2 \omega^2}{d_2^2 - d_3^2}}, & \textrm{ if } |\omega| < \frac{d_2^2 - d_3^2}{d_3 c_3},\\
      d_3 + c_3 |\omega|, & \textrm{ otherwise.}
    \end{cases}
  \end{align*}
\end{lmm}

\begin{proof}
  Under the assumption of $D_2$-covariance, take without loss of
  generality  $\vec{c} =  (0,0,c_3)^T$.   Then  without loss  of
  generality we take  $d_2 \ge d_1$ and $c_3 \ge  0$.  If $c_3 =
  0$ without loss of generality we also take $d_3 \ge d_2$.

  First notice  that $\vec{y}^*$,  which attains the  maximum in
  Eq.~\eqref{eq:qubit-threshold},  lies   in  the   $yz$  plane.
  Indeed,  any  ellipse  obtained  as the  intersection  of  the
  ellipsoid $|D^{-1} \vec{y}|_2  \le 1$ and a  plane containing the
  $z$ axis  is, up to  a $z$ rotation,  a subset of  the ellipse
  obtained as the intersection of the ellipsoid $|D^{-1} \vec{y}|_2
  \le 1$ and the $yz$ plane.

  The generic vector on the boundary  of the $yz$ ellipse can be
  parametrized as
  \begin{align*}
    \vec{y} = \left(0, \pm d_2 \sqrt{ 1 - \frac{z^2}{d_3^2}}, z
    \right)^T,
  \end{align*}
  with  $z \in  [-d_3,  d_3]$, and  thus  the maximum  Euclidean
  distance in Eq.~\eqref{eq:distance} is given by
  \begin{align}
    \label{eq:distance2}
    \Delta(\omega) = \max_{z \in [-d_3,d_3]}  \sqrt{ d_2^2 \left( 1 -
        \frac{z^2}{d_3^2} \right) + (z - \omega c_3)^2 }.
  \end{align}
  By explicit computation one has
  \begin{align*}
    & \frac{d  \Delta(\omega)}{dz}  \\ = & \left[  d_2^2  \left(  1  -
        \frac{z^2}{d_3^2}   \right)  +   (z   - \omega   c_3)^2
    \right]^{-\frac12}  \left[   \left(  1  -\frac{d_2^2}{d_3^2}
      \right) z - c_3 \omega \right],
  \end{align*}
  which  is zero  for  $z^* =  \frac{c_3  d_3^2 \omega}{d_3^2  -
    d_2^2}$, and
  \begin{align*}
    \frac{d^2  \Delta(\omega)}{dz^2} \Big|_{z  =  z^*} =  \left[
      d_3^2 \sqrt{d_2^2 \left( 1 + \frac{c_3^2 \omega^2}{d_2^2 -
            d_3^2}\right)}\right]^{-1} (d_3^2 - d_2^2),
  \end{align*}
  namely $z^*$  attains the maximum  in Eq.~\eqref{eq:distance2}
  whenever $d_2 \ge d_3$.  Therefore  the maximum is attained by
  $z =  z^*$ iff $-d_3 <  z^* \le d_3$, namely  when $|\omega| <
  \frac{d_2^2  -  d_3^2}{d_3  c_3}$,  and   by  $z  =  \pm  d_3$
  otherwise.    By   replacing   $z^*$    and   $\pm   d_3$   in
  Eq.~\eqref{eq:distance2} the statement follows.
\end{proof}

\setcounter{thm}{0}

\begin{thm}
  Any given  binary conditional probability distribution  $p$ is
  compatible  with  any   given  qubit  $D_2$-covariant  channel
  $\ch{X}$ if and only if
  \begin{align*}
    \max_{\omega  \in   \Omega}  (p  \cdot   w^\pm(\omega)  -
    \thr{\ch{X}}{w^\pm(\omega)}) \le 0,
  \end{align*}
  where $\Omega  := \{ 0, \pm  \omega_1, \pm \omega_2, \pm  1 \}
  \cap [-1,1]$.
\end{thm}

\begin{proof}
  The  function  $f^\pm(\omega)  := p^T  \cdot  w^\pm(\omega)  -
  \thr{\ch{X}}{\omega}$ is  the minimum of  continuous functions
  $g^\pm(\omega)    :=     p^T    \cdot     w^{\pm}(    \omega)-
  \frac12(1+|\omega|)$ and $h^\pm(\omega)  := p^T \cdot w^{\pm}(
  \omega)- \frac12(1+\Delta(\omega))$.  Therefore, $\max_{\omega
    \in  [-1,1]}  f^\pm(x)$  is  attained  by  those  values  of
  $\omega$ maximizing $g^\pm(\omega)$  or $h^\pm(\omega)$, or in
  the intersections of $g^\pm(\omega)$ and $h^\pm(\omega)$.

  The function $g^\pm(\omega)$ is  piece-wise linear and attains
  its maximum on $[-1, 1]$ in $0$.  The function $h^\pm(\omega)$
  is quasi-concave continuous with a continuous derivative. Indeed
  \begin{align*}
    & 2 \frac{d h^\pm(\omega)}{d\omega} \\ = & \begin{cases}
      \pm(p_{1|1}-p_{2|2})- \frac{d_2 c_3^2 \omega}{\sqrt{(d_2^2-d_3^2)(d_2^2 - d_3^2 + c_3^2 \omega^2)}} , & \textrm{ if } |\omega| < \frac{d_2^2 - d_3^2}{d_3 c_3},\\
      \pm(p_{1|1}-p_{2|2})- \sgn(\omega) c_3,  & \textrm{ otherwise},\\
    \end{cases}
  \end{align*}
  is continuous and
  \begin{align*}
    2 \frac{d^2 h^\pm(\omega)}{d\omega^2} = \begin{cases}
      - \frac{d_2 c_3^2 (d_2^2 - d_3^2)^2  }{\left[ (d_2^2-d_3^2)(d_2^2 - d_3^2 + c_3^2 \omega^2)\right]^{3/2}}, & \textrm{ if } |\omega| < \frac{d_2^2 - d_3^2}{d_3 c_3},\\
      0, &  \textrm{ if  } |\omega|  > \frac{d_2^2  - d_3^2}{d_3
        c_3}.
    \end{cases}
  \end{align*}
  is non positive. Therefore $h^\pm(\omega)$ attains its maximum
  on $[-1, 1]$ in $0$, $\pm  1$, or in the zero $\pm\omega_1$ of
  its first derivative.

  Due  to the  piece-wise linearity  of $g^\pm(\omega)$  and the
  quasi-concavity  of   $h^\pm(\omega)$,  since   $g^\pm(0)  \ge
  h^\pm(0)$  and  $g^\pm(\pm1)  \le h^\pm(\pm1)$  one  has  that
  $g^\pm(\omega)$ and  $h^\pm(\omega)$ intersect in  exactly two
  points $\pm\omega_2 \in [-1, 1]$, thus the statement follows.
\end{proof}

\setcounter{cor}{1}

\begin{cor}
  Any given  binary conditional probability distribution  $p$ is
  compatible with the amplitude-damping channel $\ch{A}^\lambda$
  if and only if
  \begin{align*}
    \left(    \sqrt{p_{1|2}p_{2|1}}   -    \sqrt{p_{1|1}p_{2|2}}
    \right)^2 \le \lambda.
  \end{align*}
\end{cor}

\begin{proof}
  One has $c_{3}= 1 - \lambda$, $d_2 = \sqrt{\lambda}$, and $d_3
  = \lambda$, thus
  \begin{align*}
    \omega_{1}                                                 =
    \sqrt{\frac{\lambda}{(1-\lambda)((1-\lambda)-(p_{1|1}-p_{2|2})^{2}
        )}} (p_{1|1}-p_{2|2}),
  \end{align*}
  and $\omega_2  = 1$.  By explicit  computation, the conditions
  $\omega_1  \in   \mathbb{R}$  and   $|\omega_1|  \le   1$  are
  equivalent  to  $(p_{1|1} -  p_{2|2})^2  <  1 -  \lambda$  and
  $(p_{1|1}  - p_{2|2})^2  \le (1  - \lambda)^2$,  respectively,
  thus  $\omega_1  \in  [-1,1]$  is equivalent  to  $(p_{1|1}  -
  p_{2|2})^2 \le (1 - \lambda)^2$ for any $\lambda > 0$.

  By      explicit     computation,      the     maximum      in
  Eq.~\eqref{eq:qubit-compatibility}  is attained  at $\omega  =
  \pm\omega_1$ and $\omega = \pm1$ whenever $|p_{1|1} - p_{2|2}|
  \le 1  - \lambda$  and $|p_{1|1}  - p_{2|2}|  > 1  - \lambda$,
  respectively. Thus Eq.~\eqref{eq:qubit-compatibility} becomes
  \begin{align*}
    |  p_{1|1}   -  p_{1|2}  |   -  \sqrt{\lambda  \left[   1  -
        \frac{(p_{1|1} -p_{2|2})^2}{1-\lambda} \right]} \le 0,
  \end{align*}
  whenever  $|p_{1|1} -  p_{2|2}| \le  1 -  \lambda$, which,  by
  solving  in  $\lambda$,  becomes $\lambda_-  \le  \lambda  \le
  \lambda_+$  whenever $\lambda  \le 1  - |p_{1|1}  - p_{2|2}|$,
  where $\lambda_\pm = (\sqrt{p_{1|1} p_{2|2}} \pm \sqrt{p_{1|2}
    p_{2|1}})^2$.   By  explicit  computation $1  -  |p_{1|1}  -
  p_{2|2}| \le \lambda_+$, so the statement follows.
\end{proof}


\begin{thebibliography}{}
\bibitem{selftest} As a comparison we notice that, while our
  approach is  top-down, i.e. it aims  at characterizing the
  set of correlations compatible with a given hypothesis, in
  self-testing~\cite{MY03,  MMMO06, BLMMS09,  MYS12, SASA15,
    WWS16}  the  approach is  bottom-up,  i.e.   it aims  at
  characterizing  the set  of hypotheses  compatible with  a
  given correlation.
\bibitem{circuit}  The  role  of   the  circuit  within  any
  hypothesis is to describe  the space-time structure of the
  experiment,  usually assumed  to obey  special relativity.
  Thus,   while   circuits   corresponding   to   space-like
  correlations   are   constrained   by   the   no-signaling
  principle, those  corresponding to  time-like correlations
  are    only   constrained    by   the    strictly   weaker
  no-signaling-from-the-future        principle~\cite{Dar07,
    Oza07}.  As a consequence, the hypotheses falsifiable in
  a time-like  test are inherently more  specific than those
  falsifiable in  a space-like  test: for instance,  while a
  Bell-test~\cite{Bel64,   CHSH69,  Cir80}   can  rule   out
  classical  theory  altogether,  a classical  model  always
  exists supporting any given time-like correlation.
\end{thebibliography}
\end{document}